\documentclass[aps,prl,showpacs,twoside,twocolumn,superscriptaddress]{revtex4}
 
\usepackage{amsmath,amsfonts,amssymb,color,epsfig,graphics,graphicx,latexsym,revsymb,theorem,verbatim}

\usepackage{amsmath,amsfonts,amssymb,caption,hyperref,color,epsfig,graphics,graphicx,latexsym,mathrsfs,revsymb,theorem,url,verbatim,epstopdf,cleveref}

\usepackage{hyperref}

\usepackage{appendix}
\hypersetup{colorlinks,linkcolor={blue},citecolor={blue},urlcolor={red}}

\newtheorem{definition}{Definition}
\newtheorem{proposition}[definition]{Proposition}
\newtheorem{lemma}[definition]{Lemma}

\newtheorem{theorem}[definition]{Theorem}
\newtheorem{corollary}[definition]{Corollary}
\newtheorem{conjecture}[definition]{Conjecture}

\newtheorem{remark}[definition]{Remark}
\newtheorem{example}[definition]{Example}

\def\squareforqed{\hbox{\rlap{$\sqcap$}$\sqcup$}}
\def\qed{\ifmmode\squareforqed\else{\unskip\nobreak\hfil
\penalty50\hskip1em\null\nobreak\hfil\squareforqed
\parfillskip=0pt\finalhyphendemerits=0\endgraf}\fi}
\def\endenv{\ifmmode\;\else{\unskip\nobreak\hfil
\penalty50\hskip1em\null\nobreak\hfil\;
\parfillskip=0pt\finalhyphendemerits=0\endgraf}\fi}
\newenvironment{proof}{\noindent \textbf{{Proof.~} }}{\qed}

\def\bcj{\begin{conjecture}}
\def\ecj{\end{conjecture}}
\def\bcr{\begin{corollary}}
\def\ecr{\end{corollary}}
\def\bd{\begin{definition}}
\def\ed{\end{definition}}
\def\bea{\begin{eqnarray}}
\def\eea{\end{eqnarray}}
\def\bem{\begin{enumerate}}
\def\eem{\end{enumerate}}
\def\bex{\begin{example}}
\def\eex{\end{example}}
\def\bim{\begin{itemize}}
\def\eim{\end{itemize}}
\def\bl{\begin{lemma}}
\def\el{\end{lemma}}
\def\bpf{\begin{proof}}
\def\epf{\end{proof}}
\def\bpp{\begin{proposition}}
\def\epp{\end{proposition}}
\def\bqu{\begin{question}}
\def\equ{\end{question}}
\def\br{\begin{remark}}
\def\er{\end{remark}}
\def\bt{\begin{theorem}}
\def\et{\end{theorem}}
\def\bma{\begin{bmatrix}}
\def\ema{\end{bmatrix}}

\def\btb{\begin{tabular}}
\def\etb{\end{tabular}}
\newcommand{\nc}{\newcommand}

\def\a{\alpha}
\def\b{\beta}
\def\g{\gamma}

\def\ve{\varepsilon}

\def\r{\rho}

\def\G{\Gamma}

\def\L{\Lambda}

\def\Ps{\Psi}

 \nc{\bA}{{\bf A}} \nc{\bB}{{\bf B}} \nc{\bC}{{\bf C}}
 \nc{\bD}{{\bf D}} \nc{\bE}{{\bf E}} \nc{\bF}{{\bf F}}
 \nc{\bG}{{\bf G}} \nc{\bH}{{\bf H}} \nc{\bI}{{\bf I}}
 \nc{\bJ}{{\bf J}} \nc{\bK}{{\bf K}} \nc{\bL}{{\bf L}}
 \nc{\bM}{{\bf M}} \nc{\bN}{{\bf N}} \nc{\bO}{{\bf O}}
 \nc{\bP}{{\bf P}} \nc{\bQ}{{\bf Q}} \nc{\bR}{{\bf R}}
 \nc{\bS}{{\bf S}} \nc{\bT}{{\bf T}} \nc{\bU}{{\bf U}}
 \nc{\bV}{{\bf V}} \nc{\bW}{{\bf W}} \nc{\bX}{{\bf X}}
 \nc{\bZ}{{\bf Z}}

\nc{\cA}{{\cal A}} \nc{\cB}{{\cal B}} \nc{\cC}{{\cal C}}
\nc{\cD}{{\cal D}} \nc{\cE}{{\cal E}} \nc{\cF}{{\cal F}}
\nc{\cG}{{\cal G}} \nc{\cH}{{\cal H}} \nc{\cI}{{\cal I}}
\nc{\cJ}{{\cal J}} \nc{\cK}{{\cal K}} \nc{\cL}{{\cal L}}
\nc{\cM}{{\cal M}} \nc{\cN}{{\cal N}} \nc{\cO}{{\cal O}}
\nc{\cP}{{\cal P}} \nc{\cQ}{{\cal Q}} \nc{\cR}{{\cal R}}
\nc{\cS}{{\cal S}} \nc{\cT}{{\cal T}} \nc{\cU}{{\cal U}}
\nc{\cV}{{\cal V}} \nc{\cW}{{\cal W}} \nc{\cX}{{\cal X}}
\nc{\cZ}{{\cal Z}}

\nc{\hA}{{\hat{A}}} \nc{\hB}{{\hat{B}}} \nc{\hC}{{\hat{C}}}
\nc{\hD}{{\hat{D}}} \nc{\hE}{{\hat{E}}} \nc{\hF}{{\hat{F}}}
\nc{\hG}{{\hat{G}}} \nc{\hH}{{\hat{H}}} \nc{\hI}{{\hat{I}}}
\nc{\hJ}{{\hat{J}}} \nc{\hK}{{\hat{K}}} \nc{\hL}{{\hat{L}}}
\nc{\hM}{{\hat{M}}} \nc{\hN}{{\hat{N}}} \nc{\hO}{{\hat{O}}}
\nc{\hP}{{\hat{P}}} \nc{\hR}{{\hat{R}}} \nc{\hS}{{\hat{S}}}
\nc{\hT}{{\hat{T}}} \nc{\hU}{{\hat{U}}} \nc{\hV}{{\hat{V}}}
\nc{\hW}{{\hat{W}}} \nc{\hX}{{\hat{X}}} \nc{\hZ}{{\hat{Z}}}

 \nc{\bbA}{\mathbb{A}} \nc{\bbB}{\mathbb{B}} \nc{\bbC}{\mathbb{C}} \nc{\bbR}{\mathbb{R}} \nc{\bbF}{\mathbb{F}}

\nc{\gu}{{\frak u}}

\def\dim{\mathop{\rm Dim}}

\def\tr{\mathop{\rm Tr}}

\def\SU{{\mbox{\rm SU}}}
\def\SO{{\mbox{\rm SO}}}

\def\Ort{{\mbox{\rm O}}}
\def\Un{{\mbox{\rm U}}}

\def\dg{\dagger}

\def\ox{\otimes}

\def\ra{\rightarrow}

\newcommand{\bra}[1]{\langle#1|}
\newcommand{\ket}[1]{|#1\rangle}
\newcommand{\proj}[1]{| #1\rangle\!\langle #1 |}
\newcommand{\ketbra}[2]{|#1\rangle\!\langle#2|}

\begin{document}
\title{The existence of distinguishable bases in three-dimensional subspaces of qutrit-qudit systems under one-way local operations and classical communication}

\author{Zhiwei Song}\email[]{zhiweisong@buaa.edu.cn }
\affiliation{School of Mathematics and Systems Science, Beihang University, Beijing 100191, China}

\author{Lin Chen}\email[]{linchen@buaa.edu.cn (corresponding  author)}
\affiliation{School of Mathematics and Systems Science, Beihang University, Beijing 100191, China}
\affiliation{International Research Institute for Multidisciplinary Science, Beihang University, Beijing 100191, China}

\def\Dbar{\leavevmode\lower.6ex\hbox to 0pt
{\hskip-.23ex\accent"16\hss}D}
\author {{ Dragomir {\v{Z} \Dbar}okovi{\'c}}}\email[]{dragomir@rogers.com}
\affiliation{Department of Pure Mathematics and Institute for
Quantum Computing, University of Waterloo, Waterloo, Ontario, N2L
3G1, Canada} 


\begin{abstract}
We show that every three-dimensional subspace of qutrit-qudit complex or real systems has a distinguishable basis under one-way local operations and classical communication (LOCC). In 
particular this solves an open problem proposed in [J. Phys. A, 40, 7937, 2007]. We construct a three-dimensional space whose locally distinguishable basis is unique and apply the uniqueness property to the task of state transformation. We also construct a three-dimensional locally distinguishable multipartite space assisted with entanglement. On the other hand, we show that four-dimensional indistinguishable bipartite subspaces under one-way LOCC exist. Further, we show that the environment-assisted classical capacity of every channel with a three-dimensional environment is at least $\log_2 3$, and the environment-assisting classical capacity of any 
qutrit channel is $\log_2 3$. We also show that every two-qutrit state can be converted into a generalized classical state near the quantum-classical  boundary by an entanglement-breaking channel.  
\end{abstract}

\date{ \today }




\maketitle 

\textit{Introduction.-} Quantum nonlocality has been widely regarded as one of the fundamental properties of quantum mechanics. Nonlocality can be manifested by Bell inequalities \cite{1964Bell}, quantum entanglement \cite{PhysRev47777}, the indistinguishability of multipartite states under local operations and classical communications (LOCC) \cite{PhysRevA59,PhysRevLett96,PhysRevLett109,PhysRevLett100}, the construction of uniform states in heterogeneous systems \cite{shi2022k} and strongly nonlocal unextendible product bases \cite{shi2022strongly}. These physical phenomena and resources have been the key ingredients in applications such as quantum computing, data hiding and secret sharing. 
Further, the 
asymptotic LOCC
discrimination of 
bipartite states is 
related to the 
Chernoff distance \cite{cmp2009}.

It is known that every two-dimensional subspace of arbitrary multipartite system has a locally distinguishable orthonormal basis \cite{2001Local}. The same conclusion holds for every qubit-qudit subspace \cite{2011Any}. 
In contrast, finding such a basis in higher-dimensional spaces turns out to be a hard problem. Watrous proved the existence of bipartite subspaces having no distinguishable basis under LOCC \cite{watrous2005bipartite}. The currently known minimal dimension of a bipartite indistinguishable subspace under LOCC is seven \cite{duan2009distinguishability}.
The main problem so far is 
whether every three-dimensional subspace of 
qutrit-qudit systems has a basis which is distinguishable under one-way LOCC \cite{2007On}. 
Note that one-way LOCC requires a fixed ordering
of the actions by the parties. More specifically, suppose Alice and Bob share a combined quantum system, we say that an orthonormal basis is  distinguishable under one-way LOCC if measurements are made first by Alice and the results are sent to Bob through classical communications, who can select a measurement to distinguish the basis.  A subspace is called {\em indistinguishable under one-way LOCC} if such a basis does not exist.
The problem mentioned above has attracted attention in the past years \cite{chen2018locally,kribs2019quantum}.

In this paper, by using some well-known facts from differential topology, we give a positive answer to the problem stated above, but a negative answer if the dimension of the subspace becomes four. We show that the locally distinguishability of bipartite system can be extended to multipartite system assisted with entanglement. We also investigate the uniqueness property of the locally distinguishable basis and apply the property to the task of state transformation.

Locally distinguishable and indistinguishable subspaces play an 
important role in the study of classical corrected capacity of quantum channels, which is defined as the best classical capacity one can achieve when the receiver
of the noisy channel can be assisted with a friendly environment
through LOCC \cite{2005Correcting}. In a short word, measurements on the environment and system can help to recover the input information from the channel. The classical corrected capacity of any quantum channel is at least one bit of information according to the result of \cite{2001Local}. It was later shown that the classical corrected capacity of any rank-two quantum channel is $\log_2 d$, where $d$ is the dimension of the input space \cite{2011Any}. On the other hand, there exist quantum channels with classical corrected capacity less than $\log_2 d$ \cite{watrous2005bipartite}.
Our results imply that by first measuring on the environment, the environment-assisted classical capacity of every channel with a three-dimensional environment is at least $\log_2 3$. In addition, by first measuring on the system, the environment-assisting classical capacity of every qutrit channel is $\log_2 3$. Next, we apply our result to investigate the nonlocality without entanglement in terms of the so-called generalized classical state in the quantum-classical boundary \cite{chen2011detecting}. We show that every two-qutrit state can be converted into a generalized classical state near the quantum-classical  boundary by a local entanglement-breaking channel from the side of system Alice (or Bob). Further, if the other side, namely the system Bob (or Alice) is in 
the maximally mixed state then the generalized classical state becomes a classical state. So the quantumness of both systems can be removed by a local operation by Alice (or Bob) only. The full classicality also implies the tasks of deterministic local broadcasting \cite{chen2011detecting}. 

\textit{Description and proof of the K-M conjecture.-} Let $\cH=\cH_A\ox\cH_B$ be the bipartite Hilbert space with 
$\dim \cH_A=m,\dim \cH_B=n$.
Let $\cM$ denote the space of Hermitian operators on 
$\cH$. We partition each $M\in\cM$ into $m^2$ square blocks $M_{ij}$ of order $n$.
For each $M\in \cM$, we set $M_A:=\sum_{j=1}^n \bra{j}_B M \ket{j}_B$ and 
$M_B:=\sum_{j=1}^m \bra{j}_A M \ket{j}_A$ where $\{\ket{j}_A\}$ 
and $\{\ket{j}_B\}$ are the computational bases of 
$\cH_A$ and $\cH_B$ respectively.

We introduce two vector subspaces of $\cM$, namely 
$\cM_0:=\{M\in\cM:M_B=0\}$ and $\cM_{00}=\{M\in\cM_0: M_A=0\}$.
Note that
$M_B=\sum_{i=1}^m M_{ii}.$ We refer to the $M_{ii}$ as the {\em diagonal blocks} of $M$. We say that $M\in\cM$ is a {\em dd-matrix} if all diagonal blocks $M_{ii}$ are diagonal matrices. 
Denote by $\cD_{00}$ the subspace of $\cM_{00}$ consisting of all dd-matrices.
We refer to a quantum state as a {\em dd-state} if the density matrix of the state is a dd-matrix. 
The inequality $M\ge 0$ means that $M$ is positive semidefinite. 

By $\Un(n)$ we denote the group of unitary matrices of order $n$, 
and we refer to the direct product $\Un(m)\times\Un(n)$ as the 
{\em local unitary group}. The matrices in $\Un(n)$ with 
determinant 1 form the special unitary group $\SU(n)$. The subgroup 
$\cM \SU(n)$ consists of all monomial matrices in $\SU(n)$. (A matrix is monomial 
if each row and each column has exactly one nonzero entry 
and all nonzero entries have modulus 1.)
The action of $\Un(m) \times \Un(n)$ on $\cM$ is defined by 
\begin{eqnarray}
(U,V)\cdot M := (U\ox V)M(U\ox V)^\dag.
\end{eqnarray}
Since the center of $\Un(m) \times \Un(n)$ acts trivially on $\cM$, the orbits of $\Un(m) \times \Un(n)$ in $\cM$ are the same as those of $\SU(m)\times\SU(n)$, and we refer to them as 
$LU-orbits$. Two matrices $M_1,M_2\in\cM$ are 
{\em LU-equivalent} if they belong to the same LU-orbit, i.e., 
$(U,V)\cdot M_1=M_2$ for some $(U,V)\in\SU(m)\times\SU(n)$.

Assume now that $m=n=3$. Thus for $M\in\cM$ we have 
$M=[M_{ij}]$, $i,j=1,2,3$. Let $\cM_1:=\{M\in\cM:M_B=I_3\}$. 
The following conjecture was proposed by King and Matysiak \cite{2007On}. Let us state their mathematical formulation of the conjecture to which we shall refer as the {\em K-M conjecture}.

\textbf{{\em K-M conjecture:}} If $M\in\cM_1$ and $M\ge0$, then $M$ is LU-equivalent to a dd-matrix.

Now we present the main result of this paper. The detailed proof will be given in Appendices A and B.
\begin{theorem}
	\label{3dim}
	The K-M conjecture is true. Equivalently, any three-dimensional subspace of $\bbC^3\otimes \bbC^n$ has an orthonormal basis which is distinguishable under one-way LOCC.
\end{theorem}

For the equivalence of the mathematical formulation above and the original one in terms of local distinguishability see \cite{2007On}. Here we point out that any positive semidefinite matrix $M\in \cM_1$ 
can be associated with a three-dimensional subspace of $\bbC^3\otimes \bbC^n$ with a chosen orthonormal basis and a specified measurement basis of $\bbC^3$. Further, changing the orthonormal basis of the subspace induces a map $M\rightarrow (I_3, V)\cdot M$ for some $V\in \SU(3)$ and changing the measurement basis 
induces a map $M\rightarrow (U,I_3)\cdot M$ for some $U\in \SU(3)$.

The method used in the proof of Theorem \ref{3dim} can be also applied to the real case. Thus we obtain the following result, a detailed proof of which will be given in 
Appendix C (Theorem \ref{thm:real}).
\begin{theorem}
Any three-dimensional subspace of $\bbR^3\otimes \bbR^n$ has an orthonormal basis which is distinguishable under one-way LOCC.
\end{theorem}

\textit{Construction and application of unique locally distinguishable bases in multipartite systems.-} 
We construct a three-dimensional subspace of $\bbC^3\otimes \bbC^n$ whose one-way locally distinguishable basis is unique (ignoring the phase factors). As a contrast, any subspace of $\bbC^2\otimes \bbC^n$ has infinitely many one-way locally distinguishable bases \cite{2011Any}. 

To be specific,
let
\begin{eqnarray}
G=\left(
\begin{array}{ccccccccc}
\frac{2}{3} & 0 & 0 & 0 & 0 & 0 & 0 & 0 & 0 \\
0 & \frac{1}{\sqrt{3}} & 0 & 0 & 0 & \frac{i}{3 \sqrt{3}} & \frac{2}{3 \sqrt{3}} & 0 & 0 \\
0 & 0 & \frac{\sqrt{2}}{3} & 0 & 0 & 0 & 0 & 0 & 0 \\
0 & 0 & 0 & \frac{1}{\sqrt{3}} & 0 & 0 & 0 & 0 & 0 \\
0 & 0 & 0 & 0 & \frac{\sqrt{2}}{3} & 0 & 0 & 0 & 0 \\
0 & 0 & 0 & 0 & 0 & \frac{\sqrt{11}}{3\sqrt{3}} & \frac{3+2 i}{3\sqrt{33}} & 0 & 0 \\
0 & 0 & 0 & 0 & 0 & 0 & \frac{1}{\sqrt{33}} & 0 & 0 \\
0 & 0 & 0 & 0 & 0 & 0 & 0 & \frac{2}{3} & 0 \\
0 & 0 & 0 & 0 & 0 & 0 & 0 & 0 & \frac{1}{\sqrt{3}} \\
\end{array}
\right),
\end{eqnarray}
and
\begin{eqnarray}
\label{eq:psi1}
\notag
&&\ket{\Psi_1}=\ket{1}\otimes \ket{r_1}+\ket{2}\otimes \ket{r_4}+\ket{3}\otimes \ket{r_7},\\
\notag
&&\ket{\Psi_2}=\ket{1}\otimes \ket{r_2}+\ket{2}\otimes \ket{r_5}+\ket{3}\otimes \ket{r_{8}},\\
&&\ket{\Psi_3}=\ket{1}\otimes \ket{r_3}+\ket{2}\otimes \ket{r_6}+\ket{3}\otimes \ket{r_{9}},
\end{eqnarray}
where $\ket{r_i}$ is the $i$-th column vector of $G$. 
A calculation shows
$G^\dagger G=\frac{1}{9}H_0+\frac{1}{3}I_9$ where 
$H_0$ is given in (\ref{c-P}).
Hence $G^\dagger G$ is a dd-matrix and $(G^\dagger G)_B=I_3$. This implies that $\{\ket{\Psi_1},\ket{\Psi_2},\ket{\Psi_3}\}$ is a distinguishable basis under one-way LOCC.
Using Lemma \ref{H0m} in Appendix B, one can verify that  $(U, V)\cdot (G^\dagger G)$ is a dd-matrix only if $U,V\in \cM \SU(3)$. We obtain that the subspace spanned by $\ket{\Psi_i}$'s contains no other locally distinguishable basis. Using any multipartite state $\ket{\a}$, one can further construct the multipartite orthonormal basis $\ket{\Ps_1}\otimes\ket{\a},\ket{\Ps_2}\otimes\ket{\a},\ket{\Ps_3}\otimes\ket{\a}$, which is a unique one-way locally distinguishable basis in the multipartite space.

It is straightforward to see that the three states $\ket{\Psi_1}$, $\ket{\Psi_2}$, and $\ket{\Psi_3}$ are equivalent under local unitary equivalence, and each of them has entanglement approximately 1.53 ebits of the von Neumann entropy of $\tr_A\proj{\Ps_1}$. Further, by calculation one can see that
 every state in the span of the three states has entanglement more than 1.52 ebits. Hence, the entanglement of formation (EOF) of every mixed state whose range is contained in the span is also more than 1.52 ebits \cite{bennett1996mixed}. We don't know whether this EOF is exactly equal to that of the pure state $\ket{\Psi_1}$.

Next, we apply the uniqueness property of some one-way locally distinguishable bases to the task of state transformation under LU-equivalence \cite{bennett2000exact}. We say that two $n$-partite states $\a$ and $\b$ are LU-equivalent when there is a LU gate $U=\otimes^n_{j=1}U_j$ such that $\a=U\b U^\dg$. The LU-equivalent states have common properties useful for quantum-information processing, because they can be locally prepared from each other. Due to the great number of parameters, it is usually not easy to determine whether two states are LU-equivalent. Due to the uniqueness property proven by $U,V\in \cM \SU(3)$, one can see that the non-normalized bipartite state $G^\dg G$ 
is not LU-equivalent to any two-qutrit dd-state $\rho$, which is LU-equivalent to another dd-state via $U, V \not \in \cM\SU(3)$. 
Such a state $\rho$ can be chosen as $\rho=[M_{ij}]$ with all blocks being the diagonal matrices, e.g., the generalized classical states \cite{chen2011detecting}. By a similar reason, determining the LU-equivalence of two multipartite states is likely when one of them has a bipartite reduced density operator with the uniqueness property.   

\textit{Discrimination of multipartite spaces.-} The above results can be extended to multipartite systems assisted with entanglement. Taking the tripartite system as an example, suppose Alice, Bob and Charlie share a combined system of $\bbC^3\otimes \bbC^m\otimes \bbC^n$. By Theorem \ref{3dim}, any three-dimensional subspace $\cP$ has an orthonormal basis $\ket{\eta_1},\ket{\eta_2},\ket{\eta_3}$ written as  $\ket{\eta_i}=\sum_{j=1}^3 \ket{a_{j}}_A\otimes \ket{\rho_{i,j}}_{BC}$, where $\ket{a_{j}}$ is an orthonormal basis of  $A$-system and, for each fixed $j$, the three bipartite states $\ket{\rho_{i,j}}_{BC}$ in
$BC$-system are orthogonal. Alice measures her system by the basis $\{\ket{a_{1}},\ket{a_{2}},\ket{a_{3}}\}$ and tells the result to  Charlie through classical communications. 
Bob teleports his particle to Charlie by using $\log_2m$ entanglement and classical communications so that Charlie owns $\ket{\rho_{i,j}}_{BC}$.
Because the $\ket{\rho_{i,j}}_{BC}$'s are orthogonal to each other, they are distinguishable by Charlie. The above process can be extended to any number of systems as follows.

\begin{theorem}
	Suppose $A_1,\cdots,A_n$ share an $n$-partite system of $\bbC^3\otimes  \bbC^{d_2}\otimes \cdots \otimes  \bbC^{d_{n}}$ with $n\ge 2$.
Then any three-dimensional subspace is locally distinguishable assisted with $\log_2 (d_2...d_{n-1})$ ebits as the entanglement of a $d_2...d_{n-1}$-level maximally entangled state, as well as one-way classical communications from $A_1,\cdots, A_{n-1}$ to $A_n$.
\end{theorem}

\textit{Existence of bipartite four-dimensional indistinguishable subspaces under one-way LOCC.-}  
We shall prove the following theorem.
\begin{theorem}
	\label{4dim}
	There exist four-dimensional indistinguishable bipartite subspaces under one-way LOCC.
\end{theorem}
\begin{proof}
Let $\dim \cH_A=3,\dim \cH_B=4$. 
Each $M\in \cM$ is partitioned into 9 square blocks of order 4. 
Since the dimensions of the manifolds 
$\SU(3)\times\SU(4)\times\cD_{00}$
and
$\cM_{00}$
are 119 and 120 respectively, the map
$f:\SU(3)\times\SU(4)\times\cD_{00} \to \cM_{00}$ 
defined by
$f(X,Y,Z)=(X, Y)\cdot Z$
is not onto. Let us choose a matrix $K_0\in\cM_{00}$ which is
not in the image of $f$, i.e., $K_0$ is not LU-equivalent to
any dd-matrix. Next, we choose a small $\epsilon>0$ such that
$K:=\frac{1}{3}I_{12} + \epsilon K_0\ge 0$. Moreover, $K$ is not LU-equivalent to any dd-matrix.
Further, $K=P^\dag P$ for a matrix $P$ who has 12 columns.
We now define four bipartite states 
\begin{eqnarray}
	\ket{\Phi_k}=
	\ket{1}\ox\ket{p_k}+
	\ket{2}\ox\ket{p_{k+4}}+\ket{3}\ox\ket{p_{k+8}},
\end{eqnarray}
for $k=1,2,3,4$, where $\ket{p_k}$ is the $k$-th column vector of $P$. One can verify that $K_B=I_4$ and thus these four states
are orthonormal.
Since $K$ is not LU-equivalent to any dd-matrix, we deduce that the subspace spanned by $\ket{\Phi_i}'s$ is indistinguishable under one-way LOCC.
\end{proof}

\textit{Application to classical capacity of quantum channels.-}
Any quantum channel $\Phi$ can be viewed as arising from a unitary interaction $U$ between the system $\cH$ and the environment $\cE$. 
The unitary operator $U$ maps the orthogonal input states to orthogonal ones in $\cH\otimes \cE$.  
Specifically, we can write 
\begin{eqnarray}
\label{channel}
\Phi(\proj{\psi})=\text{Tr}_{\cE}[U(\proj{\psi}\otimes \proj{\epsilon})U^\dagger],
\end{eqnarray}
where $\ket{\epsilon}$ is the initial state of the environment and the partial trace is taken over the environment.  However, the output of the system may not be orthogonal after tracing out the environment, and thus cannot be distinguished perfectly.
It is possible to more reliably distinguish output states of a noisy quantum
channel by using measurements on the environment or on the system. This idea of enhancing the channel corrected capacity has been considered in a number of settings \cite{wer, 2005Correcting, winter2005environment}. The two notions {\em environment-assisted} and {\em environment-assisting} were introduced in \cite{winter2005environment}.
Specifically, assume that the input state $\ket{\psi}$ in (\ref{channel}) varies over the system space $\cH$, the state $U(\ket{\psi}\otimes \ket{\epsilon})$  varies over a subspace $\cV$ of $\cH\otimes \cE$, where $\cE$ denotes the environment space. Suppose $\cV$ has a basis  $U(\ket{\psi_i}\otimes \ket{\epsilon})$ that can be distinguished using one-way LOCC, then we can encode classical information in the system states $\ket{\psi_i}$ and completely recover the information by measuring the environment (resp. system) and followed by a selected measurement on the system (resp. environment). In this setting, we say that the classical environment-assisted (resp. environment-assisting) capacity of the channel is $\log_2 d$, where $d$ is the dimension of $\cH$. 
For a channel with a three-dimensional environment, the dimension of $\cE$ is three. For a qutrit channel, the dimension of $\cH$ is three.
The results stated in the next corollary follow from
Theorem \ref{3dim}. 
\begin{corollary}
The environment-assisted classical capacity of every channel with a three-dimensional environment is at least $\log_2 3$. The environment-assisting classical capacity of any qutrit channel is $\log_2 3$. 
\end{corollary}

\textit{Application to quantum-classical  boundary.-} The study of nonlocality without entanglement in terms of correlations such as discord has attracted attention in the past decades \cite{2001Hend,2010dakic,2011luo}. The nonlocality can be manifested by the 
multipartite states lying near the quantum-classical boundary, namely the so-called  generalized classical and classical states \cite{chen2011detecting}. For convenience we shortly review them as follows. Let 
$\{\ket{\phi(\vec{i})}\}=\{\ket{\phi^{(1)}_{i_1}\phi^{(2)}_{i_2}\dots\phi^{(N)}_{i_N}}\}$ be a basis of product states. It is known that a multipartite separable state $\r$ is a convex sum of product states. We refer to $\r$ as a \textit{generalized classical (resp. classical) state for the $k^{th}$ system} if $\r$ is diagonal in a product state basis and the $\ket{\phi^{(k)}_{i_k}}$ are linearly independent (resp. orthonormal). Further, 
the state $\rho$ is \textit{fully generalized classical (resp. fully classical)}  
if it is diagonal in every system with a linearly independent (resp. orthonormal) basis. The generalized classical and classical states can be both efficiently detected by using existing semidefinite programming (SDP) \cite{chen2011detecting}. Given any two-qutrit state $\a$ acting on a system $\cH_{AB}$, by using Theorem \ref{3dim} we deduce that there exists an orthonormal basis $\{\ket{a_j}\}$ of system $A$ such that $\alpha=\sum^3_{i,j=1}\ket{a_i}\bra{a_j}\otimes M_{ij}$ where $M_{11}, M_{22},M_{33}$ are simultaneously congruent to diagonal matrices. Hence
there exists an entanglement breaking channel $\L:\a\ra \b=\sum^3_{j=1}(\proj{a_j}\otimes I_3)\a(\proj{a_j}\otimes I_3)$, such that $\b$ is a classical state w.r.t. the system $A$ and a generalized classical state w.r.t. the system $B$. In particular if $\a_B$ is the maximally mixed state then $\b$ becomes a fully classical state. 
In this sense, the quantumness of both of two distant systems can be totally removed by a local operation of only one system. Further, the full classicality is related to the tasks of deterministic local broadcasting and deterministic non-disruptive local state identification \cite{chen2011detecting}. In addition, the length of a generalized classical state equals its rank, and it represents the minimum cost of generating the state \cite{chen2013dimensions}. A similar argument may be extended to the multipartite system.   

\textit{Summary and outlook.-} We have shown that a distinguishable basis under one-way LOCC exists in any three-dimensional subspace of qutrit-qudit real or complex systems. 
We have applied the results to quantum information issues such as the locally distinguishability of multipartite spaces, the corrected capacity of quantum channels, and the quantum-classical boundary.
There are several questions arising from this letter. The first question is how to construct analytical expressions for such a basis. 
Another question is whether any bipartite three-dimensional subspace is distinguishable under one-way LOCC. We conjecture that the answer is yes. This would imply that the environment-assisted classical capacity of every qudit quantum channel is at least $\log_2 3$.
The corresponding mathematical conjecture asserts the following:
if $M\ge 0$ is a matrix of order $3m$ $(m\ge 4)$, partitioned into $m^2$ blocks $M_{ij}$ of order 3, with $M_B=I_3$, then $M$ is LU-equivalent to a dd-matrix. Finally, we have shown the existence of bipartite four-dimensional indistinguishable subspaces under one-way LOCC. However, the question: does such an indistinguishable subspace exist under a general LOCC protocol, remains open? 

We thank Li Yu and Nengkun Yu for useful comments. ZWS and LC were supported by the NNSF of China (Grant No. 11871089). 

\bibliographystyle{unsrt}
\bibliography{km}
\clearpage 
\onecolumngrid
\appendix
\section{\label{AppA}APPENDIX A: The proof of K-M conjecture}
We first prove the following Lemma.
\bl
\label{cj:main}
Each of the following three assertions is equivalent to 
the K-M conjecture: 

(i) each matrix in $\cM_0$ is LU-equivalent to a dd-matrix;

(ii) each matrix in $\cM_1$ is LU-equivalent to a dd-matrix;

(iii) each matrix in $\cM_{00}$ is LU-equivalent to a dd-matrix.

\el
\bpf
{\em K-M conjecture} $\Longrightarrow$ (i): Let $M_0\in\cM_0$ be arbitrary.  Choose small $\ve>0$ such that 
$M_1:=I_9/3+\ve M_0>0$. 
By the {\em K-M conjecture} $M_1$ is LU-equivalent to a dd-matrix. Hence, 
$M_0=\ve^{-1}(M_1-I_9/3)$ is also LU-equivalent to a dd-matrix.

(i) $\Longrightarrow$ (ii): This is true because 
$\cM_1=\cM_0+I_9/3$. 

The implications
(ii) $\Longrightarrow$ {\em K-M conjecture} and 
(i) $\Longrightarrow$ (iii) are trivial.  

(iii) $\Longrightarrow$ (ii): Let $M\in\cM_0$ be arbitrary. Then 
$N:=M-M_A\ox I_3/3 \in \cM_{00}$. By (iii), $N$ is LU-equivalent to a dd-matrix. Hence, the same is true for $M=N+M_A\ox I_3/3$.
\epf

\textbf {Proof of K-M conjecture}
	
	In view of Lemma \ref{cj:main}, it suffices to prove that each matrix in $\cM_{00}$ is LU-equivalent to a dd-matrix. Recalling the definition of $\cD_{00}$, we have to prove that the map
	\begin{eqnarray}
	\label{cf}
	f: \SU(3)\times \SU(3)\times \cD_{00}\rightarrow \cM_{00}
	\end{eqnarray}
	defined by $f(U,V,H)=(U,V)\cdot H$ is onto. 
	Note that $\SU(3) \times \SU(3) \times \cD_{00}$ and $\cM_{00}$ are smooth manifolds of dimensions $8+8+52=68$ and 64, respectively. Let $\cP\cD_{00}$ and $\cP\cM_{00}$ denote the real projective spaces associated with $\cD_{00}$ and $\cM_{00}$, respectively. For nonzero $H \in \cM_{00}$, we denote by $[H]$ the 1-dimensional subspace of $\cM_{00}$ viewed as a point of $\cP\cM_{00}$.
	Since $\cD_{00}$ and $\cM_{00}$ are real vector spaces and $f$ is linear in $H$, it induces a smooth map
	\begin{eqnarray} \label{c-phi}
	\notag
	&&\phi: \SU(3) \times \SU(3) \times \cP\cD_{00} \to \cP\cM_{00}, \\
	&&\phi(U,V,[H]):=[f(U,V,H)]=[(U,V)\cdot H].
	\end{eqnarray}
	To prove the theorem, it suffices to show that $\phi$ is onto. 
	
	The manifolds $\SU(3) \times \SU(3) \times \cP\cD_{00}$ and $\cP\cM_{00}$ are compact with no boundary and have dimensions 67 and 63, respectively. 
		We denote by $\Gamma$ the subgroup $\cM \SU(3)\times \cM \SU(3)$
	of $\SU(3)\times \SU(3)$.
	For convenience we shall write $\g\in\G$ as the ordered
	pair $(\g_1,\g_2)$ where $\g_1,\g_2\in\cM \SU(3)$.
	One can verify that the subspace $\cD_{00} \subseteq \cM_{00}$ is $\Gamma$-invariant. Consequently $\G$ acts on the 
	manifold $\SU(3) \times \SU(3) \times \cD_{00}$ as follows: 
	\begin{eqnarray}
	\g\bullet(U,V,H):=
	(U\g_1^\dag,V\g_2^\dag,(\g_1,\g_2)\cdot H).
	\end{eqnarray}
	This action of $\G$ induces an action on the manifold 
	$\SU(3) \times \SU(3) \times \cP\cD_{00}$ which we will denote by the same symbol. Thus we have
	\begin{eqnarray}
	\g\bullet(U,V,[H]):=
	(U\g_1^\dag,V\g_2^\dag,[(\g_1,\g_2)\cdot H]).
	\end{eqnarray}
	
	It is straightforward to verify that this action of $\G$ is free, which means that if 
	$\g\in\G$ fixes a point $(U,V,[H])$ then $\g=(I_3,I_3)$. The corresponding quotient space (also known as the orbit space) 
	$\cN:=( \SU(3) \times \SU(3) \times \cP\cD_{00} )/\G$
	is also a smooth compact manifold, see e.g.
	\cite[Appendix II, Proposition 2, p.229]{MR206917520040101} or \cite[Appendix II, pp.173-196]{MR192913620020101}. 
	We shall denote by $(U,V,[H])^{\#}$ the image in $\cN$ of a point $(U,V,[H])\in\SU(3)\times\SU(3)\times\cP\cD_{00}$. 
	Further, the dimension of $\cN$ is $63$. Since the map 
	$\phi$ is smooth and constant on each $\Gamma$-orbit, it induces a smooth map 
	\begin{eqnarray}
	\notag
	\label{cphin}
	&&\phi^{\#}: \cN \to \cP\cM_{00},\\
	&&\phi^{\#}((U,V,[H])^{\#}):=[(U,V)\cdot H].
	\end{eqnarray}
	
	Let $P=(U_0,V_0,[H_0]) \in \SU(3) \times \SU(3) \times 
	\cP\cD_{00}$, 
	where
	\begin{eqnarray} \label{c-P}
	U_0=V_0=\frac{1}{8} \bma 6 & 4 & 2\sqrt{3} \\
	-1 & 6 & -3\sqrt{3} \\
	-3\sqrt{3} & 2\sqrt{3} & 5 \ema,~ 
	H_0=
	\bma 1&0&0&0&0&0&0&0&0\\
	0&0&0&0&0&i&2&0&0\\
	0&0&-1&0&0&0&0&0&0\\
	0&0&0&0&0&0&0&0&0\\
	0&0&0&0&-1&0&0&0&0\\
	0&-i&0&0&0&1&1&0&0\\
	0&2&0&0&0&1&-1&0&0\\
	0&0&0&0&0&0&0&1&0\\
	0&0&0&0&0&0&0&0&0
	\ema,
	\end{eqnarray}
	and let $Q:=\phi(P)=[(U_0,V_0)\cdot H_0]$. 
	
	There is a nice parametrization \cite{PhysRevD.38.1994} of $\SU(3)$ in terms of 8 angles:
	$\theta_i$ (i=1,2,3) and $\phi_i$ (i=1,2,\ldots,5).
	Considering the map $f$ defined in (\ref{cf}), we used 16  angles (eight for each copy of $\SU(3)$) as coordinates and computed the rank of the Jacobian matrix of $f$ at the point 
	$(U_0,V_0,H_0)$. This rank is 64, and so this point is a regular point of $f$. Consequently, $P$ is a regular point of $\phi$, and $P^{\#}$ is a regular point of $\phi^{\#}$. 
	
	We claim that the equation 
	$\phi^{\#}((U,V,[H])^{\#})=Q$ has only one solution, namely 
	$P^{\#}$. Equivalently, we claim that all
	solutions of the equation $\phi(U,V,[H])=Q$ are exactly all
	points of the $\G$-orbit of $P$.
	
	The equation $\phi(U,V,[H])=Q$ can be written as 
	$[(U,V)\cdot H]=[(U_0,V_0)\cdot H_0]$. By multiplying $H$ by
	a suitable nonzero real number, we may assume that the
	equality $(U,V)\cdot H=(U_0,V_0)\cdot H_0$ holds. Finally
	we can rewrite this equation as 
	$(U_0^\dag U, V_0^\dag V)\cdot H=H_0$. By using Lemma \ref{H0m}
	in Appendix B, we deduce that $\g_1:=U_0^\dag U$ 
	and $\g_2:=V_0^\dag V$ are monomial matrices. It follows that
	$(U,V,[H])=(U_0\g_1,V_0\g_2,[(\g_1^\dag,\g_2^\dag)\cdot H_0])\in\G\bullet P$. Thus our claim is proved.

	This means that there is only one point on the manifold $\cN$ which satisfies the equation
	$\phi^{\#}(\G\bullet (U,V,[H]))=[(U_0,V_0)\cdot H_0]$. Hence
	the point $[(U_0,V_0)\cdot H_0]$ is a regular value of $\phi^{\#}$. Consequently the map $\phi^{\#}$ is onto, see  \cite[Corollary 3.9.6, p. 96]{1993Differentiable} or 
	\cite[Chapter 2, Section 3]{MR034878119740101}. 
	This implies that $\phi$ is also onto.

\section{APPENDIX B: The proof of a necessary lemma  supporting Theorem \ref{3dim}}
\label{app:prmo}

Our objective here is to prove the following lemma 
which was used in the proof of Theorem \ref{3dim}. 
Let us recall that $P:=(U_0,V_0,[H_0])\in \SU(3)\times\SU(3)\times\cP\cD_{00}$ 
	where the matrices $U_0,V_0,H_0$ are specified in \eqref{c-P}, and that $Q:=\phi(P)=[(U_0,V_0)\cdot H_0]$. 
\begin{lemma}
	\label{H0m}
Under these assumptions, the equation 
	$\phi(U,V,[H])=Q$ implies that $(U,V)\in\G$.
\end{lemma}

\begin{proof}
	A set of Hermitian matrices can be simultaneously diagonalized by a unitary matrix if and only if they commute with each other. Let us find all $U\in \SU(3)$ such that the three diagonal blocks of 
	$(U,I_3)\cdot H_0$ commute. We shall refer to such $U$ as {\em good} matrices. Since the sum of these three blocks is zero, it suffices to make the first two blocks commute.
	
	We can write any $U\in \SU(3)$ as follows \cite{PhysRevD.38.1994}, 
	\begin{eqnarray}
		\label{U}
		U:=\bma u_{11}&u_{12}&u_{13}\\
		u_{21}&u_{22}&u_{23}\\
		u_{31}&u_{32}&u_{33}
		\ema,
	\end{eqnarray}
	where
	\begin{eqnarray}
	\label{u}
	\notag
	&&u_{11}=\cos{\theta_1}\cos{\theta_2}e^{i\phi_1}, \\
	\notag
	&&u_{12}=\sin{\theta_1}e^{i\phi_3}, \\
	\notag
	&&u_{13}=\cos{\theta_1}\sin{\theta_2}e^{i\phi_4}, \\
	\notag
	&&u_{21}=\sin{\theta_2}\sin{\theta_3}e^{-i\phi_4-i\phi_5}-\sin{\theta_1}\cos{\theta_2}\cos{\theta_3}e^{i\phi_1+i\phi_2-i\phi_3},\\
	\notag
	&&u_{22}=\cos{\theta_1}\cos{\theta_3}e^{i\phi_2}, \\
	\notag
	&&u_{23}=-\cos{\theta_2}\sin{\theta_3}e^{-i\phi_1-i\phi_5}
	-\sin{\theta_1}\sin{\theta_2}\cos{\theta_3}e^{i\phi_2-i\phi_3+i\phi_4}, \\
	\notag
	&&u_{31}=-\sin{\theta_1}\cos{\theta_2}\sin{\theta_3}e^{i\phi_1-i\phi_3+i\phi_5}
	-\sin{\theta_2}\cos{\theta_3}e^{-i\phi_2-i\phi_4}, \\
	\notag
	&&u_{32}=\cos{\theta_1}\sin{\theta_3}e^{i\phi_5}, \\
	&&u_{33}=\cos{\theta_2}\cos{\theta_3}e^{-i\phi_1-i\phi_2}
	-\sin{\theta_1}\sin{\theta_2}\sin{\theta_3}e^{-i\phi_3+i\phi_4+i\phi_5},
	\end{eqnarray}
	where
	$0\le \theta_1,\theta_2,\theta_3\le \frac{\pi}{2}$, 
	$0\le \phi_1,\phi_2,\phi_3,\phi_4,\phi_5\le 2\pi$. 
	It is easy to verify that $U$ is a monomial matrix if 
	$\{\theta_1,\theta_2,\theta_3\} \subseteq \{0,\pi/2\}$.

	Denote by $R=[r_{ij}]$ the product $D_1D_2$ and by $S=[s_{ij}]$ the commutator $D_1D_2-D_2D_1$ of the first and the second diagonal blocks, $D_1$ and $D_2$ respectively, of the matrix
	$(U,I_3)\cdot H_0$. 
	Note that $U$ is good if and only if $S=0$ or, equivalently, 
	$R$ is Hermitian.
	Since $D_1$ and $D_2$ are Hermitian matrices, $S$ is skew-Hermitian. 
	One can easily verify that $s_{33}=0$. Since $S$ has trace 0, we have $s_{22}=-s_{11}$.
	
	From now on in this proof we assume that $U$ is good and
we will prove that $U\in\cM\SU(3)$.
	For convenience we set 
	$\psi:=\phi_1+\phi_2-\phi_3+\phi_4+\phi_5$.
	
	Our first claim is that at least one of the three angles $\theta_i$ is equal to 0 or $\pi/2$. A calculation shows that 
	\begin{eqnarray}
		s_{11}=\frac{3i}{2}\cos^2 \theta_1 \sin\theta_1 \sin 2\theta_2
		\sin 2\theta_3 \sin\psi.
	\end{eqnarray}
	Hence the claim holds unless $\sin\psi=0$. We may now assume 
	that $\theta_1,\theta_2 \in (0,\pi/2)$ and $\sin\psi=0$.
	Another calculation shows that 
	\begin{eqnarray}
		\notag
		&&\sin(\phi_1-\phi_4) \text{Re}(r_{12}-r_{21}) +
		\cos(\phi_1-\phi_4) \text{Im}(r_{12}+r_{21}) \\
		=&&-\frac{1}{4}\cos\theta_1 \sin 2\theta_1 \sin 2\theta_3 \cos\psi
		-\sin\theta_1 \sin 2\theta_3 (1-3\cos^2 \theta_1 
		\sin^2 \theta_2) \sin\psi.
	\end{eqnarray}
	Since $R$ is a Hermitian matrix the LHS vanishes, and since 
	$\theta_1 \in (0,\pi/2)$ and $\sin\psi=0$ it follows that 
	$\sin 2\theta_3 =0$. Hence our claim is true.
	
	Our second claim is that if two of the angles $\theta_i$ belong to $\{0,\pi/2\}$ then so does the third, and so 
	$U\in\cM\SU(3)$. There are 12 cases to consider. We shall prove that the claim holds in the case $\theta_1=\theta_2=0$. The proofs in the other cases are similar and are omitted.
	By setting $\theta_1=\theta_2=0$ in $S$, a calculation shows that
	\begin{eqnarray}
		s_{13}=-\sin 2\theta_3 \; e^{-i(\phi_1+\phi_2+\phi_5)}.
	\end{eqnarray}
	Since $U$ is good, $S=0$ and we must have 
	$\sin 2\theta_3=0$, i.e., $\theta_3$ is also 0 or $\pi/2$. Hence $U\in\cM\SU(3)$.
	
	Our third claim is that if at least one $\theta_i$ is $0$ or $\pi/2$ then $U\in\cM\SU(3)$. 
	
	There are six cases to consider: $\theta_i=0$ or $\theta_i=\pi/2, (i=1,2,3)$. We shall give the proofs for the two cases with $i=1$. We omit the proofs in the other four cases as they are similar. 
	
	Suppose first that $\theta_1=0$. 
By setting $\theta_1=0$ in $S$, a computation shows that
	\begin{eqnarray}
s_{12}=-\sin 2\theta_2 \cos 2\theta_3 e^{i(\phi_4-\phi_1)}.
	\end{eqnarray}
On the other hand, by setting $\theta_1=0$ and $\theta_3=\pi/4$ in $S$ we find that
	\begin{eqnarray}
		s_{23}=\frac{1}{2} \sin \theta_2 (i-2\cos^2 \theta_2)
		e^{-i(\phi_2+\phi_4+\phi_5)}.
	\end{eqnarray}
Since $S=0$, the above expressions for $s_{12}$ and $s_{23}$ imply that $\sin 2\theta_2=0$. Our second claim now shows that $U\in\cM\SU(3)$.
	
	Next suppose that $\theta_1=\pi/2$. By setting $\theta_1=\pi/2$ in $S$, a computation shows that
	\begin{eqnarray}
	s_{12}=\sin 2\theta_2 \cos 2\theta_3 e^{i(\phi_4-\phi_1)}
		+\cos^2\theta_2 \sin 2\theta_3 
		e^{-i(2\phi_1+\phi_2-\phi_3+\phi_5)}
		-\sin^2\theta_2 \sin 2\theta_3 
		e^{i(\phi_2-\phi_3+2\phi_4+\phi_5)}.
	\end{eqnarray}
	After multiplying by $e^{i(\phi_1-\phi_4)}$, we obtain that
	\begin{eqnarray}
	\sin 2\theta_2 \cos 2\theta_3 
	+\sin 2\theta_3 (e^{-i\psi} \cos^2 \theta_2 
				 -e^{i\psi} \sin^2 \theta_2) =0.
				 	\end{eqnarray}
	By taking the imaginary parts in this equation, we 
obtain that $\sin 2\theta_3 \sin\psi =0$. 
If $\sin 2\theta_3 =0$ then our second claim implies that
$U\in\cM\SU(3)$. We may assume that $\sin\psi=0$, 
and so $\cos \psi=\pm 1$. By taking the real parts in the above equation, we obtain that
	\begin{eqnarray}
		\sin 2\theta_2 \cos 2\theta_3 + \cos 2\theta_2 \sin 2\theta_3
		\cos\psi=0. 
	\end{eqnarray}
Suppose $\cos\psi=1$, then one of $\sin(\theta_2+\theta_3)$ and $\cos(\theta_2+\theta_3)$ equals to 0. Suppose $\cos\psi=-1$, then one of $\sin(\theta_2-\theta_3)$ and $\cos(\theta_2-\theta_3)$ equals to 0.
In both cases it is easy to deduce from (\ref{u}) that 
 $U\in\cM\SU(3)$. 
 	
	It remains to prove that $V \in \cM \SU(3)$.
	Since $U$ is good and monomial, the three diagonal blocks of 
	$(U,I_3)\cdot H_0$ are just a permutation of the three diagonal blocks of $H_0$. From (\ref{c-P}), we see that each of the three diagonal blocks of $(U,I_3)\cdot H_0$ has 3 distinct eigenvalues. Consequently the three diagonal blocks of $(U,V)\cdot H_0$ are diagonal matrices only if $V\in\cM\SU(3)$. This completes the proof.
\end{proof}

\section{APPENDIX C: The real K-M conjecture and its proof}
\label{prreal}

Here we consider the real analog of the {\em K-M conjecture}. Thus $\cH_A$ and $\cH_B$ will be real Hilbert spaces of dimension 3, and their tensor product 
$\cH=\cH_A\ox\cH_B$ is taken over the reals, $\bbR$. 
In this section, $\cM$ denotes the space of real symmetric matrices of order 9. 

The subspaces $\cM_1$, $\cM_0$ and $\cM_{00}$ of $\cM$ are defined in the same way as in the complex case. 
The local unitary group $\Un(3)\times\Un(3)$ is now replaced by the local orthogonal (LO) group $\Ort(3)\times\Ort(3)$. Its action on $\cM$ is given by the formula 
$(A,B)\cdot M:=(A\ox B)M(A\ox B)^T$. Since the matrices 
$(\pm I_3,\pm I_3)$ act trivially on $\cM$, each 
$\Ort(3)\times\Ort(3)$-orbit in $\cM$ is also an 
$\SO(3)\times\SO(3)$-orbit. We refer to these orbits as the 
{\em LO-orbits}. Two matrices in $\cM$ are {\em LO-equivalent} if they belong to the same LO-orbit. 

Given a bipartite matrix $M$, we denote its partial transpose w.r.t. system $B$ by $M^{\G_B}$, i.e., if 
$M=\sum^{m_1}_{i=1}\sum^{n_1}_{j=1}\ketbra{i}{j}\ox M_{i,j}$
then
$M^{\G_B}=\sum^{m_1}_{i=1}\sum^{n_1}_{j=1}\ketbra{i}{j}\ox M_{i,j}^T$. 
We introduce an additional subspace, $\cM_2$. It is the subspace of $\cM_{00}$ consisting of all matrices $H$ such that 
$H^{\G_B}=H$. Equivalently, $\cM_2$ consists of all matrices 
$H\in\cM_{00}$ having all blocks $H_{ij}$ symmetric. One can easily verify that all four subspaces: 
$\cM_1$, $\cM_0$, $\cM_{00}$ and $\cM_2$ are LO-invariant.

The following conjecture is the real analog of the original 
{\em K-M conjecture}. We shall refer to it as the 
{\em real K-M conjecture}.

\bcj
In the real setting described above, if $M\in\cM_1$ and $M\ge0$, then $M$ is LO-equivalent to a dd-matrix.
\ecj

As in the complex case, there are several equivalent formulations of this conjecture. We state four of them in the next lemma. 

\bl \label{cj:real-00} Each of the following assertions is 
equivalent to the real K-M conjecture:

(i) each matrix in $\cM_0$ is LO-equivalent to a dd-matrix;

(ii) each matrix in $\cM_1$ is LO-equivalent to a dd-matrix;

(iii) each matrix in $\cM_{00}$ is LO-equivalent to a dd-matrix;

(iv) each matrix in $\cM_2$ is LO-equivalent to a dd-matrix.
\el

\bpf
In the cases (i), (ii)  and (iii) the equivalence can be proved in the same way as in the complex case. 
Obviously (iii) implies (iv). To prove the converse, assume that (iv) holds. Let 
$N\in\cM_{00}$ and set $M=N+N^{\G_B}$. Since $M^{\G_B}=M$, (iv) implies that there exists $(A,B)\in\Ort(3)\times\Ort(3)$ such that $(A,B)\cdot M=(A,B)\cdot N+(A,B)\cdot N^{\G_B}$ is a 
dd-matrix. Since 
$(A,B)\cdot N^{\G_B}=\left( (A,B)\cdot N \right)^{\G_B}$, 
the matrices $(A,B)\cdot N$ and $(A,B)\cdot N^{\G_B}$ share the same diagonal blocks. As their sum is a dd-matrix, the same is true for $(A,B)\cdot N$. Thus (iv) implies (iii).
\epf

Denote by $\cD_2$ the subspace of $\cM_2$ consisting of all 
dd-matrices in $\cM_2$.

\begin{theorem}
\label{thm:real}
	The real K-M conjecture is true.
\end{theorem}
\begin{proof} We have to show that the map
	\begin{equation} \label{funkcija-f}
		f: \SO(3) \times \SO(3) \times \cD_2 \to \cM_2
	\end{equation}
	defined by $f(X,Y,Z)=(X,Y)\cdot Z$ is onto. 
	Since $\cD_2$ and $\cM_2$ are real vector spaces and 
	the map $f$ is linear in $Z$, $f$ induces a map
	\begin{equation} \label{funkcija-phi}
		\phi: \SO(3) \times \SO(3) \times \cP\cD_2 \to \cP\cM_2,
	\end{equation}
	where $\cP\cD_2$ and $\cP\cM_2$ are the real projective spaces
	associated with $\cD_2$ and $\cM_2$, respectively.
	It suffices to show that the map $\phi$ is onto. Note that 
	$\SO(3) \times \SO(3) \times \cP\cD_2$ and $\cP\cM_2$ are compact smooth manifolds with empty boundaries. Moreover they have the same dimension, namely 24, and the map $\phi$ is smooth. For nonzero $Z \in \cM_2$ we denote by $[Z]$ the 
	1-dimensional subspace of $\cM_2$ viewed as a point of 
	$\cP\cM_2$.
	
	Denote by $S_4$ the subgroup of $\SO(3)$ consisting of all
	monomial matrices in $\SO(3)$. It is isomorphic to the symmetric 
	group of order 24. Let us also introduce the subgroup 
	$\Gamma:=S_4 \times S_4$ of $\SO(3) \times \SO(3)$. One can easily verify that the subspace $\cD_2 \subseteq \cM_2$ is $\Gamma$-invariant. Consequently, we can define an action of  $\Gamma$ on the 
	manifold $\SO(3) \times \SO(3) \times \cP\cD_2$ as follows:
	$(\gamma_1,\gamma_2)\bullet (X,Y,[Z]):=
	(X\gamma_1^T,Y\gamma_2^T,[(\gamma_1,\gamma_2)\cdot Z]).$
	It is easy to verify that this action is free. Hence 
	each orbit of $\Gamma$ has exactly $|\Gamma|=576~ (=24^2)$ points.
	We deduce that the quotient space 
	$\cN:=( \SO(3) \times \SO(3) \times \cP\cD_2 )/\Gamma$
	is also a smooth compact manifold of dimension 24, see e.g.
	\cite[Appendix II, Proposition 2, p.229]{MR206917520040101}.
	Moreover the map $\phi$ induces a smooth map 
	$\phi^{\#}: \cN \to \cP\cM_2$.
	
	Let $P=(X_0,Y_0,[Z_0]) \in 
		\SO(3)\times\SO(3)\times\cP\cD_2$ and let 
		$Q:=f(X_0,Y_0,Z_0)\in\cM_2$
	where
\begin{eqnarray} \label{tacka-P}
		X_0=Y_0=\bma 1&0&0\\0&0&-1\\0&1&0\ema,
		Z_0=
		\bma 0&0&0&0&0&1&0&0&0\\
		0&1&0&0&0&0&0&0&1\\
		0&0&-1&1&0&0&0&1&0\\
		0&0&1&0&0&0&0&1&0\\
		0&0&0&0&1&0&1&0&0\\
		1&0&0&0&0&-1&0&0&0\\
		0&0&0&0&1&0&0&0&0\\
		0&0&1&1&0&0&0&-2&0\\
		0&1&0&0&0&0&0&0&2\ema.
	\end{eqnarray}
	We used six Euler angles (three for each copy of $\SO(3)$) as coordinates and computed the Jacobian determinant of $f$ at the point $(X_0,Y_0,Z_0)$. It is not 0, and so this point is a regular point of the map $f$. Consequently, the point $P$ is a regular point of $\phi$. Further, the image of 
	$P$ in $\cN$, i.e. the $\Gamma$-orbit of $P$ 
	is a regular point of $\phi^{\#}$.
	
	We shall now prove that $\G\bullet P$ is the 
	unique point of $\cN$ which satisfies the equation
	$\phi^{\#}(\G\bullet (X,Y,[Z]))=[Q]$. This will 
	prove that the point $[Q]$ is a regular 
	value of $\phi^{\#}$ and so the map $\phi^{\#}$ must be onto, see  \cite[Corollary 3.9.6, p. 96]{1993Differentiable} or 
	\cite[Chapter 2, Section 3]{MR034878119740101}. 
	Hence $\phi$ will be onto too.

It is immediate from the definition of the $\G$-action that
the set $f^{-1}(Q)$ is $\G$-invariant. In particular we have
$\G\bullet (X_0,Y_0,Z_0)\subseteq f^{-1}(Q)$.

Let $(X_1,Y_1,Z_1)\in f^{-1}(Q)$ be arbitrary. Our first claim is that $X_1\in S_4$. We have 
\begin{equation} \label{jed-f}
(X_1,Y_1)\cdot Z_1 = (X_0,Y_0)\cdot Z_0
\end{equation} 
which we can rewrite as 
\begin{equation}
(I_3,Y_0^T Y_1)\cdot Z_1=(X_1^T X_0,I_3)\cdot Z_0.
\end{equation} 
Let $X_2:=X_1^T X_0$ and $M:=(X_2,I_3)\cdot Z_0$. We partition $M$ into nine $3$ by $3$ blocks $M_{ij}$. 
Since $X_0\in S_4$, it suffices to prove that $X_2\in S_4$.

The equation displayed above implies that $M$ can be transformed into a dd-matrix by the action of $\SO(3)$ on the $B$-system only.
Therefore the three diagonal blocks $M_{ii}$ of $M$ must commute with each other. In particular the commutator 
$S:=M_{11}M_{22}-M_{22}M_{11}=[s_{ij}]$ must vanish. Since the blocks $M_{ii}$ are symmetric matrices, $S$ is skew-symmetric.

Let us write the matrix $X_2\in \SO(3)$ as a 
	function of the three Euler angles 
	\begin{equation} \label{mat-X2}
X_2=	\bma 
 	  \cos\a \cos\g -\sin\a \cos\b \sin\g &
      -\cos\a \sin\g -\sin\a \cos\b \cos\g &
	  \sin\a \sin\b \\
 	  \sin\a \cos\g +\cos\a \cos\b \sin\g &
      -\sin\a \sin\g +\cos\a \cos\b \cos\g &
	 -\cos\a \sin\b \\
 	  \sin\b \sin\g & \sin\b \cos\g & \cos\b \\
	\ema,
	\end{equation}
where $\a,\g\in[0,2\pi]$ and $\beta\in[0,\pi]$. Then we compute the matrix $M$ and the commutator $S$. One can easily verify that the solutions $(\a,\b,\g)$ 
of this system in which $\b\in\{0,\pi/2,\pi\}$ give all 24 
matrices $X_2$ in $S_4$, and nothing else. Thus we may assume that $\sin\b \cos\b \ne 0$. 

The three equations $s_{12}=0$, $s_{13}=0$, $s_{23}=0$ can be written in the following form: 	
\begin{eqnarray}
	\label{prva}
	&&2\cos 2\a \cos\b \cos\g (3-2\cos^2\g)
	+\sin 2\a \sin\g 
	(2\cos^2\b \cos^2\g -5\cos^2\b +2\cos^2\g +1)=0, \\
	\label{druga}
	&&2\cos 2\a \cos 2\b \sin 2\g 
	+\sin 2\a \cos\b 
	(6\cos^2 \b \cos^2 \g -5\cos^2 \b -2\cos^2 \g +3)=0, \\
	\label{treca}
	&&2\cos 2\a \cos\b \sin\g (1+\cos^2\g)
	+\sin 2\a \cos\g
	(\cos^2\b \cos^2\g +3\cos^2\b +\cos^2\g -2)=0.
	\end{eqnarray}
We have omitted the factor $\sin\b$ from \eqref{prva} and 
\eqref{treca}. If $\sin\g=0$ then it is easy to see that our claim follows from the above equations. Hence we assume
from now on that $\sin\g\ne0$.
	
Since each of the three equations above is linear and homogeneous in 
$\sin 2\a$ and $\cos 2\a$, it follows that the matrix

\begin{eqnarray}
\bma
\cos\b \cos\g (3-2\cos^2\g) & 
\sin\g(2\cos^2\b\cos^2\g-5\cos^2\b +2\cos^2\g+1)\\
\cos 2\b \sin 2\g & 
\cos\b(6\cos^2\b\cos^2\g-5\cos^2\b-2\cos^2\g+3) \\
\cos\b\sin\g(1+\cos^2\g) &
\cos\g(\cos^2\b\cos^2\g+3\cos^2\b+\cos^2\g-2) \\
\ema
\end{eqnarray}
must have rank 1. By equating to 0 the minor not containing
the middle row, and recalling that $\cos\b\ne0$, we obtain the 
formula
\begin{eqnarray}
\cos^2\b=\frac{1+2\sin^2 2\g}{4+\sin^2\g+2\sin^2 2\g}.
\end{eqnarray}
Finally, by equating to 0 the minor not containing the first
row (and using this formula) we obtain the equation
\begin{eqnarray}
2+3\cos^2\g+10\cos^4\g-14\cos^6\g=0.
\end{eqnarray}
As the LHS is equal to 
$1+\sin^2\g+\sin^2 2\g+14\sin^2\g\cos^4\g$,
this equation has no solutions and our claim is proved, i.e. 
$X_1\in S_4$.

	The proof that $Y_1\in S_4$ is similar and is omitted. Note that in any solution $(X_1,Y_1,Z_1)$ of 
	\eqref{jed-f} the matrices $X_1$ and $Y_1$ determine 
$Z_1$ uniquely. Consequently $f^{-1}(Q)$ is a single
$\G$-orbit.
	This implies that there is only one point 
	of the manifold $\cN$ which satisfies the equation
	$\phi^{\#}(\Gamma\bullet (X,Y,[Z]))=[(X_0,Y_0)\cdot Z_0]$. Hence
	the point $[(X_0,Y_0)\cdot Z_0]$ is a regular 
	value of $\phi^{\#}$ and thus the map $\phi^{\#}$ must be onto.
	Consequently $\phi$ is also onto.
\end{proof}

\end{document}